\documentclass[12pt]{article}
\usepackage{amsfonts,amssymb,amsmath,amsthm}
\usepackage{graphicx,caption}
\usepackage[mathscr]{euscript}

\def\Z{\mathbb{Z}}

\def\p{ \partial }

\def\bq{ \begin{equation} }
\def\eq{ \end{equation} }
\def\ben{ \begin{eqnarray} }
\def\en{ \end{eqnarray} }

\def\frac#1#2{{#1\over #2}}
\def\on#1#2{\mathop{\vbox{\ialign{##\crcr\noalign{\kern2pt}
$\scriptstyle{#2}$\crcr\noalign{\kern2pt\nointerlineskip}
\kern-2pt$\hfil\displaystyle{#1}\hfil$\crcr}}}\limits}

\newtheorem{prop}{Proposition}
\newtheorem{theor}{Theorem}
\newtheorem{Defi}{Definition}

\newtheorem{remark}{Remark}
\newtheorem{cor}{Corollary}

\newcommand{\todo}[1][\null]{\ensuremath{\clubsuit}}

\begin{document}

\title{Defining differential equations for modular forms and Jacobi forms}
\author{S. Opanasenko$^1$, E.V. Ferapontov$^{1, 2}$}
   \date{}
\vspace{-20mm}
   \maketitle
\vspace{-7mm}
\begin{center}
${^1}$ Department of Mathematical Sciences \\
Loughborough University \\
Loughborough, Leicestershire LE11 3TU, UK \\[1ex]
    \ \\
$^2$Institute of Mathematics, Ufa Federal Research Centre\\
Russian Academy of Sciences, 112 Chernyshevsky Street \\
Ufa 450008, Russian Federation\\
e-mails: \\
\texttt{S.Opanasenko@lboro.ac.uk}\\
\texttt{E.V.Ferapontov@lboro.ac.uk}
\end{center}

\medskip

\begin{abstract}
It is well known that every modular form~$f$ on a discrete subgroup $\Gamma\leqslant \textrm{SL}(2, \mathbb R)$ satisfies a third-order nonlinear ODE that expresses algebraic dependence of the functions~$f$, $f'$, $f''$ and~$f'''$. These ODEs are automatically invariant under the Lie group $\textrm{SL}(2, \mathbb R)$, which acts on the solution spaces thereof with an open orbit (and the discrete stabiliser~$\Gamma$ of a generic solution). Similarly,
every modular form satisfies a fourth-order nonlinear ODE that is invariant under the Lie group $\textrm{GL}(2, \mathbb R)$ acting on its solution space with an open orbit. ODEs for modular forms can be compactly expressed in terms of the differential invariants of these actions. The invariant forms of both ODEs define plane algebraic curves naturally associated with every modular form; the corresponding ODEs can be seen as modular parametrisations of the associated curves.

After reviewing examples of nonlinear ODEs satisfied by classical modular forms (such as Eisenstein series, modular forms on congruence subgroups of level two and three, theta constants, and some newforms of weight two), we generalise these results to Jacobi forms; these satisfy involutive third-order PDE systems
that are invariant under the Lie group $\textrm{SL}(2, \mathbb R)\ltimes H$ where $H$ is the Heisenberg group.
\end{abstract}



\section{Introduction}\label{ODEMF:sec:intro}

Modular forms on a discrete subgroup $\Gamma\leqslant \textrm{SL}(2, \mathbb R)$ are holomorphic functions~$f(\tau)$ on the upper half-plane~$\mathcal H$
that satisfy the modular transformation property
\[
f\left(\frac{a\tau +b}{c\tau+d}\right)=(c\tau+d)^kf(\tau),\qquad \left(\begin{array}{cc} a&b\\c&d\end{array}\right)\in \Gamma,
\]
where $k$ is the weight of a modular form; we refer to~\cite{123} for a general theory.
It is well known that every modular form satisfies a third-order nonlinear ODE that expresses algebraic dependence of the functions $f$, $f'$, $f''$ and~$f'''$,
which is shown in, e.g., \cite[Theorem~2]{Resnikoff1966} based on the prior idea of~\cite{Hurwitz1889}
that two meromorphic functions on a Riemann surface are algebraically dependent.
What was not explicitly noted however is that every such ODE is automatically
$\mathrm{SL}(2,\mathbb R)$-invariant and can be represented by an algebraic relation
\[
F(I_k, J_k)=0
\]
where $I_k$ and $J_k$ are differential invariants of a certain $\mathrm{SL}(2,\mathbb R)$-action,
\begin{gather*}
\displaystyle I_k=\frac{k ff''-(k+1)f'^2}{f^{2+\frac4k}},
\quad
\displaystyle J_k=\frac{k^2f^2f'''-3k(k+2)ff'f''+2(k+1)(k+2)f'^3}
{f^{3+\frac6k}},
\end{gather*}
see Section~\ref{sec:invar} for their expressions in terms of the Rankin--Cohen brackets.
Thus, there is a plane algebraic curve $C\colon F(I_k, J_k)=0$ naturally associated with every modular form (note that relation $F$ depends on the modular form $f$).
The case of particular interest in number theory is where~$f$ is a newform of weight $k=2$ on a congruence subgroup $\Gamma_0(N)$.
In this case, modular functions~$I_2$ and~$J_2$ provide a modular parametrisation of~$C$.
For special values of~$N\geq 11$, one obtains elliptic curves over $\mathbb{Q}$ with a modular parametrisation; these curves $C$ have the same $j$-invariants as the elliptic curves associated with the modular form $f(\tau)$ via  the Taniyama-Shimura-Weil conjecture,
see e.g. \cite{BK, Zagier1984, Zagier1985} and Examples of Section~\ref{sec:ex} for some explicit formulae and further discussion.

Similarly, every modular form~$f(\tau)$ of weight~$k$ satisfies a fourth-order ODE, compare with~\cite[Remark, p.~342]{Resnikoff1966},
\[
{\cal F}(P_k, Q_k)=0,
\]
where $P_k$ and $Q_k$ are differential invariants of the order~$3$ and~$4$ of a certain $\mathrm{GL}(2,\mathbb R)$-action,
\begin{equation*}
\begin{array}{c}
\displaystyle P_k=\frac{(k^2f^2f'''-3k(k+2)ff' f''+2(k+1)(k+2)f'^3)^2}
{(kff''-(k+1)f'^2)^3},\\
\ \\
\displaystyle Q_k=\frac{f^2\left(k(k{+}1)ff''''-4(k{+}1)(k{+}3)f' f'''+3(k{+}2)(k{+}3)f''^2\right)}{(kff''-(k+1)f'^2)^2},
\end{array}
\end{equation*}
see Section~\ref{sec:invar} for their expressions in terms of the Rankin--Cohen brackets.
Thus, there is another plane algebraic curve ${\cal C}: {\cal F}(P_k, Q_k)=0$ naturally associated with every modular form (the relation $\cal F$ depends on the modular form $f$).
Note that the invariants $P_k$ and $Q_k$ are modular functions for every $k$, and provide a modular parametrisation of $\cal C$.
There is a natural covering map $C \to {\cal C}$, see Remark~\ref{F(P,Q)} in Section~\ref{sec:invar}.
In all examples discussed in this paper the second curve $\cal C$ turns out to be rational (even for  congruence subgroups $\Gamma_0(N)$ of higher genus), although we have no general explanation of this fact.

The origin of ODEs for classical modular forms on~$\textrm{SL}(2, \mathbb Z)$ are the Ramanujan equations for the Eisenstein series.
In what follows, we use the notation $q=\mathrm e^{2\pi i \tau}$ and denote by prime the operator
$q\frac{d}{dq}=\frac{1}{2\pi i} \frac{d}{d\tau}$. The Eisenstein series are defined as
\[
E_k(\tau)=1-\frac{2k}{B_k}\sum_{n=1}^{\infty}\sigma_{k-1}(n)q^n,
\]
where $B_k$ are the Bernoulli numbers and $\sigma_{k-1}(n)$ denotes the sum of the $(k-1)$st powers of the positive divisors of~$n$.
Explicitly, we have
\begin{gather*}
E_2(\tau)=1-24\sum_{n=1}^{\infty}\sigma_{1}(n)q^n=1-24q-72q^2-\dots,\\
E_4(\tau)=1+240\sum_{n=1}^{\infty}\sigma_{3}(n)q^n=1+240q+2160q^2+\dots,\\
E_6(\tau)=1-504\sum_{n=1}^{\infty}\sigma_{5}(n)q^n=1-504q-16632q^2-\dots,
\end{gather*}
note that $E_2$ is only quasi-modular. The Eisenstein series $E_2$, $E_4$ and~$E_6$ satisfy the Ramanujan system of ODEs,
\begin{equation}\label{Raman}
E_2'=\frac{E_2^2-E_4}{12},\qquad E_4'=\frac{E_2E_4-E_6}3,\qquad E_6'=\frac{E_2E_6-E_4^2}2,
\end{equation}
which is invariant under the action of the Lie group $\mathrm{SL}(2,\mathbb R)$ defined as
\[
\tilde \tau=\frac{a\tau+b}{c\tau+d},\quad \tilde E_2=(c\tau+d)^2E_2+12c(c\tau+d),\quad \tilde E_4=(c\tau+d)^4E_4, \quad \tilde E_6=(c\tau+d)^6E_6.
\]
Every modular form~$f$ on~$ \textrm{SL}(2, \mathbb Z)$ is a homogeneous polynomial in $E_4$ and $E_6$.
Differentiation of~$f$ with the help of~(\ref{Raman}) gives four polynomial expressions for $f$, $f'$, $f''$, $f'''$ in terms of $E_2$, $E_4$ and~$E_6$.
The elimination of $E_2$, $E_4$, $E_6$ leads to a third-order nonlinear ODE for~$f$, which inherits $\mathrm{SL}(2,\mathbb R)$-symmetry from the Ramanujan equations.

In special cases, both (third- and fourth-order) ODEs for modular forms and their invariance properties have been discussed in the literature, see, e.g., \cite{Ablowitz2006, 123, Maier, O3}.
For instance, the modular discriminant \[\Delta=q\prod_{n=1}^{\infty}(1-q^n)^{24}=\frac1{1728}(E_4^3-E_6^2)\]
satisfies an $\mathrm{SL}(2,\mathbb R)$-invariant third-order ODE~(\cite[Proposition~3]{Resnikoff1966})
\begin{gather}
36\Delta^4\Delta'''^2-14(18\Delta \Delta''-13\Delta'^2)\Delta^2\Delta'\Delta'''+48\Delta^3\Delta''^3\nonumber\\
\qquad+285\Delta^2\Delta'^2\Delta''^2-468\Delta\Delta'^4\Delta''+169\Delta'^6+48\Delta^7=0,\nonumber
\end{gather}
as well as the more well-known van der Pol--Rankin equation, which is $\mathrm{GL}(2,\mathbb R)$-invariant fourth-order ODE,
\begin{gather}
2\Delta^3\Delta''''-10\Delta^2\Delta'\Delta'''-3\Delta^2\Delta''^2+24\Delta\Delta'^2\Delta''-13\Delta'^4=0,\nonumber
\end{gather}
possessing an extra scaling symmetry $\Delta \to \lambda \Delta$ that is not present in the third-order ODE;
see Example~1 of Section~\ref{sec:ex} for the invariant forms of both equations. Note that the fourth-order ODE is a differential consequence of the third-order ODE.

It is precisely via the differential equations $F(I_k,J_k)=0$ and $\mathcal F(P_k,Q_k)=0$ that modular forms feature in various contexts in mathematical physics.
In this paper, we discuss the invariance properties of ODEs for modular forms from the point of view of symmetry analysis of differential equations,
as well as an extension of these results to Jacobi forms.
Below we list some examples of ODEs for modular forms originating from the theory of dispersionless integrable PDEs in~3D.

\medskip

\noindent {\bf First-order integrable Lagrangians.} Paper~\cite{FKT} gives a characterisation of first-order integrable Lagrangians of the form $\int F(u_{x_1}, u_{x_2}, u_{x_3})\, \mathrm dx$. 
It was observed in~\cite{FO, CFOZ} that the corresponding Lagrangian densities~$F$ are related to Picard modular forms.
In particular, for densities of the form $F=u_{x_1}u_{x_2}f(u_{x_3})$, the integrability conditions
lead to a fourth-order $\mathrm{GL}(2, \mathbb{R})$-invariant ODE for $f(\tau)$,
\[
ff''''(ff''-2f'^2)-f^2f'''^2+2f'(ff''+4f'^2)f'''-9f'^2f''^2=0,
\]
whose general solution is the Eisenstein series $E_{1,3}(\tau)$,
\[
f(\tau)=E_{1, 3}(\tau)=\sum_{(\alpha,\beta)\in\Z^2}q^{(\alpha^2-\alpha\beta+\beta^2)}=1+6q+6q^3+6q^4+12q^7+.....,
\]
see Example~4 of Section~\ref{sec:ex} for further details.

\medskip

\noindent {\bf Hirota type equations.} Paper~\cite{FHK} studies integrability of dispersionless Hirota-type equations in 3D, $F(u_{x_ix_j})=0$, where $u(x_1, x_2, x_3)$ is a function of three independent variables, and $u_{x_ix_j}$ denote second-order partial derivatives. It was shown in \cite{CF} that the `generic' integrable Hirota master-equation is expressible via genus three theta constants. In particular, for equations of the form
\[
u_{x_3x_3} -\frac{u_{x_1x_2}}{u_{x_1x_3}}-\frac{1}{6}h(u_{x_1x_1})u_{x_1x_3}^{2}=0,
\]
the integrability conditions lead to the $\mathrm{SL}(2, \mathbb{R})$-invariant Chazy equation for $h(t)$ \cite{Pavlov},
\[
h_{ttt}+2hh_{tt}-3h_t^2=0,
\]
whose general solution is expressed in terms of the Eisenstein series $E_2$,
\[
h(t)=E_2(it/\pi)=1-24\sum_{n=1}^{\infty} \sigma_1(n)\mathrm e^{-2nt}.
\]

\medskip

\noindent {\bf Second-order quasilinear PDEs.} Paper \cite{BFT} studies integrability of 3D second-order quasilinear PDEs of the form
$\sum_{i,j}f_{ij}(u_{x_1}, u_{x_2}, u_{x_3})u_{x_ix_j}=0$ where $u(x_1, x_2, x_3)$ is a function of three independent variables. In particular, for equations of the form
\[
u_{xy}+(u_xu_yr(u_t))_t=0,
\]
 the integrability conditions result in an $\mathrm{SL}(2, \mathbb{R})$-invariant third-order ODE for $r(s)$,
\begin{equation*}
r_{sss}(r_s-r^2) -r_{ss}^2+ 4{r}^3{r_{ss}}+2r_s^3 - 6{r}^2r_s^2= 0,
\end{equation*}
which has appeared in the context  of modular forms of level two \cite{Ablowitz2006}. Its generic solution is given by the Eisenstein series
\[
r(s)=1-8\sum_{n=1}^{\infty}\frac{(-1)^nnw^n}{1-w^n}, \quad w= \mathrm e^{4 s},
\]
which is associated with the congruence subgroup $\Gamma_0(2)$ of the modular group; see Section \ref{sec:ex} for further details.

\medskip

\noindent {\bf Second-order integrable Lagrangians.} Second-order integrable Lagrangians of the form $\int F(u_{xx}, u_{xy}, u_{yy}) \, \mathrm dx\mathrm dy$ were investigated in~\cite{FPX}. Under the ansatz $F=\mathrm e^{u_{xx}}g(u_{xy}, u_{yy})$, the integrability conditions lead to the following constraints for $g(z, c)$
(we set $z=u_{xy}$, $c=u_{yy}$):
\begin{equation}\label{eqg}
\begin{array}{c}
gg_{zcc}=3g_{cc}g_{z}-2 g_{zc} g_{c},\\
gg_{zzz}=g_{z}g_{zz}+4g_{zc}g-4g_{z}g_{c},\\
gg_{ccc}=g_{c}g_{cc} +2 g_{cc} g_{zz}-2 (g_{zc})^2,\\
gg_{zzc}=2g_{z}g_{zc}-g_{c}g_{zz}+2 g g_{cc}-2(g_{c})^2.
\end{array}
\end{equation}
This over-determined system for $g$ is in involution and its generic solution can be represented in the form
\[
g(z, c)=[\Delta(ic/\pi)]^{-1/8}\theta_1(ic/\pi, z)
\]
where $\Delta$ is the modular discriminant and $\theta_1$ is the Jacobi theta function,
\begin{gather*}
\theta_1(\tau, z)=2\sum_{n=0}^{\infty} (-1)^n\mathrm e^{\pi i(n+1/2)^2\tau}\sin[(2n+1)z].
\end{gather*}

\medskip

The reason for the occurrence of modular forms in the above classification results is a remarkable construction of~\cite{OS} that parametrises broad classes of dispersionless integrable systems in 3D via generalised hypergeometric functions. For special values of the parameters
(where the monodromy group of hypergeometric system is a lattice), this parametrisation leads to integrable PDEs whose coefficients are expressible via modular forms.

The structure of the paper is as follows. In Section 2 we discuss differential equations for modular forms by first describing
a general construction of $\mathrm{SL}(2, \mathbb{R})$- and $\mathrm{GL}(2, \mathbb{R})$-invariant equations in terms of the differential invariants of the corresponding group actions (Section~\ref{sec:invar}), and then illustrating the general constructions by numerous examples of modular forms on various congruence subgroups (Section~\ref{sec:ex}). Differential systems for Jacobi forms are discussed in Section~\ref{sec:Jac} by first describing differential invariants of the Jacobi group (Section~\ref{sec:Jac0}) and  then illustrating the general theory by several examples (Section~\ref{sec:Jac1}).

\section{Differential equations for modular forms}

After reviewing differential invariants of the standard $\mathrm{SL}(2, \mathbb{R})$- and $\mathrm{GL}(2, \mathbb{R})$-actions occurring in the theory of modular forms, we provide third-order and fourth-order invariant differential equations for the Eisenstein series~$E_4$ and $E_6$,  the modular discriminant~$\triangle$, Jacobi theta constants, Eisenstein series $E_{1,3}$, some modular forms of level two, and some newforms of weight two on various congruence subgroups.

\subsection{Differential invariants}
\label{sec:invar}

Consider the group~$G_k$, which is isomorphic to~$\mathrm{SL(2,\mathbb R)}$, of point transformations on a space with coordinates $(\tau,f)$, which is relevant for modular forms of weight~$k$,
\[
\tilde \tau=\frac{a\tau+b}{c\tau+d}, \quad \tilde f=(c\tau+d)^k f,\quad \text{where}\ \begin{pmatrix} a & b \\ c & d \end{pmatrix}\in\mathrm{SL}(2,\mathbb R).
\]
Its lowest-order differential invariants are
\begin{equation}
 \label{inv}
 \begin{array}{c}
\displaystyle I_k=\frac{k ff''-(k+1)f'^2}{f^{2+\frac4k}}=\frac{[f,f]_2}{(k+1)f^{2+\frac4k}},\quad\\
 \ \\
\displaystyle J_k=\frac{k^2f^2f'''-3k(k+2)ff'f''+2(k+1)(k+2)f'^3}
{f^{3+\frac6k}}=\frac{[f,\, [f,f]_2]_1}{(k+1)f^{3+\frac6k}},
\end{array}
\end{equation}
where $[\cdot,\cdot]_i$ is the $i$th Rankin--Cohen bracket; we follow the notation of~\cite[p.~53]{123}.
Any third-order $G_k$-invariant ODE can be written in terms of these two invariants as
\begin{equation}\label{FIJ}
F(I_k, J_k)=0.
\end{equation}
Note that the Lie algebra~$\mathfrak g_k$ associated with~$G_k$ is spanned by the vector fields $\p_\tau$, $2\tau\p_\tau-kf\p_f$ and $\tau^2\p_\tau-k\tau f\p_f$.
In what follows, we will also need the easily verifiable relation
\begin{equation}\label{diff}
\frac{dI_k}{J_k}=\frac{2\pi i}{k}f^{2/k}d\tau.
\end{equation}

\begin{theor}\label{T1}
The $G_k$-action on the solution space of equation~(\ref{FIJ}) is locally transitive, that is, it possesses an open orbit.
\end{theor}

\begin{proof}
Let us write equation~(\ref{FIJ}) in the form $f_{\tau\tau\tau}=r(f,f_\tau,f_{\tau\tau})$.
The projections of the evolutionary forms%
\footnote{Every vector field on a jet space with coordinates~$(\tau,f)$ can be prolonged uniquely to the vector field on an infinite jet space
with coordinates $(\tau,f,f_\tau,f_{\tau\tau},\dots)$, which has an equivalent evolutionary vector field
$\sum_{i=0}^\infty \mathrm D_{\frac{\mathrm d^if}{\mathrm d\tau^i}}\chi\ \p_{\frac{\mathrm d^if}{\mathrm d\tau^i}}$
where $\mathrm D$ is the total derivative, see~\cite[Chapter~5]{Olver1993} for rigorous details.}
of the three vector fields spanning~$\mathfrak g_k$ to its three-dimensional solution space,
\begin{gather*}
f_\tau\p_f+f_{\tau\tau}\p_{f_\tau}+r\p_{f_{\tau\tau}},\quad
(kf+2\tau f_\tau)\p_f+((k+2)f_\tau+2\tau f_{\tau\tau})\p_{f_\tau}+((k+4)f_{\tau\tau}+2\tau r)\p_{f_{\tau\tau}},\\
(\tau^2f_\tau+k\tau f)\p_f+((k+2)\tau f_\tau+k f+\tau^2 f_{\tau\tau})\p_{f_\tau}+((k+4)\tau f_{\tau\tau}+2(k+1) f+\tau^2 r)\p_{f_{\tau\tau}},
\end{gather*}
are linearly independent. Indeed, the determinant
\begin{gather*}
\left|\begin{array}{ccc}
f_\tau        & f_{\tau\tau}               & r\\
kf+2\tau f_\tau   & (k+2)f_\tau+2\tau f_{\tau\tau}      & (k+4)f_{\tau\tau}+2\tau r\\
\tau^2f_\tau+k\tau f & (k+2)\tau f_\tau+k f+\tau^2 f_{\tau\tau} & (k+4)\tau f_{\tau\tau}+2(k+1) f_\tau+\tau^2 r
\end{array}\right|\\
=k^2ff_\tau f_{\tau\tau\tau}-k((k+4)f_\tau^2+2(k+1)f^2)f_{\tau\tau}+2(k+1)(k+2)ff_\tau^2.
\end{gather*}
is not $G_k$-invariant, and thus it can not be a left-hand side of the equation satisfied by~$f$.
\end{proof}

Consider also the group~$\mathcal G_k$, which is isomorphic to~$\mathrm{GL}(2,\mathbb R)$, of point transformations on a space with coordinates $(\tau,f)$,
\[
\tilde \tau=\frac{a\tau+b}{c\tau+d}, \quad \tilde f=(c\tau+d)^k f,\quad \text{where}\ \begin{pmatrix} a & b \\ c & d \end{pmatrix}\in\mathrm{GL}(2,\mathbb R).
\]
Its lowest-order differential invariants are
\begin{equation*}\label{inv1}
\begin{array}{c}
\displaystyle P_k=\frac{(k^2f^2f'''-3k(k+2)ff' f''+2(k+1)(k+2)f'^3)^2}
{(kff''-(k+1)f'^2)^3}=\frac{(k+1)[f,\, [f,f]_2]_1^2}{[f,f]_2^{3}},\\
\ \\
\displaystyle Q_k=\frac{f^2\left(k(k{+}1)ff''''-4(k{+}1)(k{+}3)f' f'''+3(k{+}2)(k{+}3)f''^2\right)}{(kff''-(k+1)f'^2)^2}
=\frac{12(k+1)^2f^2[f,f]_4}{(k+2)(k+3)[f,f]_2^2}.
\end{array}
\end{equation*}
Any fourth-order $\mathcal G_k$-invariant ODE can be written in terms of these two invariants as
\begin{equation}\label{FPQ}
\mathcal F(P_k, Q_k)=0.
\end{equation}
Note that the Lie algebra associated with the group~$\mathcal G_k$ is spanned by the vector fields $\p_\tau$, $\tau\p_\tau$, $\tau^2\p_\tau-k\tau f\p_f$ and~$f\p_f$.
Similarly to the proof of Theorem \ref{T1}, one can show that the $\mathcal G_k$-action on the solution space of equation~(\ref{FPQ}) is locally transitive (possesses an open orbit).

There exists a simple link between the ODEs described above. Namely, every $G_k$-invariant third-order ODEs (\ref{FIJ}) possesses, as its differential consequence, a $\mathcal G_k$-invariant fourth-order ODE (\ref{FPQ}); furthermore, every $\mathcal G_k$-invariant fourth-order ODE (\ref{FPQ}) arises in this way. These results are summarised in the two propositions below.

\begin{prop}\label{3to4}
Every $G_k$-invariant third-order ODE possesses, as a differential consequence, a $\mathcal G_k$-invariant fourth-order ODE.
\end{prop}

\begin{proof}

Consider a third-order $G_k$-invariant third-order ODE of type~(\ref{FIJ}). Let us differentiate it using the formulas
\[
I_k'=\frac{\sqrt{P_k}}{k}I_k^{\frac32}f^{\frac2k},\qquad
J_k'=\frac{k^2Q_k-6(k+2)^2}{k(k+1)}I_k^2 f^{\frac2k},
\]
which can be verified by direct calculation. This gives
\[
0=F_{I_k}I_k'+F_{J_k}J_k'=\frac{I_k^{\frac32}f^{\frac2k}}{k}\left(F_{I_k}\sqrt{P_k}+F_{J_k}\frac{k^2Q_k-6(k+2)^2}{k+1}I_k^{\frac12}\right).
\]
Using the relation $J_k=\sqrt{P_k}I_k^{\frac32}$ and eliminating~$I_k$ from the pair of equations
\begin{equation}\label{34}
F(I_k, J_k)=0, \qquad F_{I_k}\sqrt{P_k}+F_{J_k}\frac{k^2Q_k-6(k+2)^2}{k+1}I_k^{\frac12}=0,
\end{equation}
we obtain the required $\mathcal G_k$-invariant fourth-order ODE of type (\ref{FPQ}).
Note that {\it algebraic} third-order ODEs (\ref{FIJ}) produce {\it algebraic} fourth-order ODEs (\ref{FPQ}).
\end{proof}

Proposition~\ref{3to4} is a direct corollary of a more general group-theoretic result.
\begin{theor}
Let $H$ and~$G$ be groups of point transformations of the space~$(\tau,f)$ with $H$ being a proper subgroup of~$G$.
If $H$ is a symmetry group of an ODE~$\mathcal E$ for~$f(\tau)$, then there is a $G$-invariant differential consequence of~$\mathcal E$.
\end{theor}

\begin{proof}
Let the space on which the groups~$H$ and~$G$ are acting be coordinatised by $(\tau,f)$, their dimensions be $m$ and~$n$, respectively, $n>m$,
and the equation~$\mathcal E$ take the form~${F(\tau,f,\dots,f_r)=0}$ where $f_i=\frac{\mathrm d^if}{\mathrm d\tau^i}$.
The order of a lowest-order invariant~$I_{m-1}$ of~$H$ is~$m-1$, and its higher-order invariants can be obtained from~$I_{m-1}$
with the help of the $H$-invariant differentiation operator~$\mathrm D$, $I_{m+k-1}=\mathrm D^kI_{m-1}$, $k\in\mathbb N$.
Since equation~$\mathcal E$ is $H$-invariant it can be written as $\bar F_r(I_{m-1},\dots,I_r)=0$ ($r$ is necessarily greater or equal to~$m-1$).
Differentiating  $n-r$ times the equation $\bar F=0$ with the help of the operator~$\mathrm D$,
we obtain a system of $n-r+1$ equations $\bar F_{r+i}(I_{m-1},\dots,I_{m-1+i})=0$, $i=0,\dots,n-r$, on $n-r+1$ invariants~$I_k$
(we assume that $n>r$, otherwise we would have to keep differentiating).
Using the implicit function theorem, the lower-order $n-r$ invariants can be excluded, which results in a single equation containing the remaining invariants~$I_k$,
but they can be rewritten in terms of the invariants $J_l$ of~$G$.
Indeed, since $H<G$, the invariants~$J_l$ of~$G$ are invariants of~$H$ as well, and thus $J_l$ are functions of~$I_k$,
$J_i=f_i(I_{m-1},\dots,I_i)$, $i\geqslant n-1$.
Thus we have a desired $G$-invariant differential consequence of~$\mathcal E$.
\end{proof}

The generalisation of this result to systems of PDEs is straightforward.

\begin{prop}\label{4to3}
Every $\mathcal G_k$-invariant fourth-order ODE is a differential consequence of some $G_k$-invariant third-order ODE.
\end{prop}

\begin{proof}

Given a~$\mathcal G_k$-invariant fourth-order ODE, let us seek a third-order ODE (\ref{FIJ}) in the form $F(I_k,J_k)=J_k-S(I_k)=0$ where $S(I_k)$ is a function to be determined. Using the relation $J_k=\sqrt{P_k}I_k^{\frac32}$, the corresponding equations (\ref{34}) can be written as
\begin{equation}\label{map}
P_k=I_k^{-3}S^2(I_k), \qquad Q_k=6\frac{(k+2)^2}{k^2}+\frac{k+1}{k^2}I_k^{-2}S(I_k)\frac{\mathrm dS(I_k)}{\mathrm dI_k}.
\end{equation}
The substitution of these relations into the fourth-order ODE, $\mathcal F(P_k, Q_k)=0$, gives a first-order differential equation for $S(I_k)$ (whose general solution depends on one arbitrary constant). Thus, there is a one-parameter family of third-order ODEs with the required property.
\end{proof}

\begin{cor}
Every modular form of weight~$k$ satisfies a $\mathrm{GL}(2,\mathbb R)$-invariant fourth-order ODE.
\end{cor}

\begin{remark}\label{SpecificCubics}  For a modular form $f(\tau)$ of weight $k$, the corresponding $G_k$-invariant third-order ODE (\ref{FIJ}) is necessarily algebraic, thus, there is a (singular) plane algebraic curve $C: F(I_k, J_k)=0$ associated to every modular form. Formulae (\ref{inv}) provide a local parametrisation of $C$. Note that this parametrisation is not necessarily modular since the denominators in~(\ref{inv}) are not modular forms in general.
In the particularly interesting case where~$f$ is a modular form of weight~$k=2$ on a suitable congruence subgroup $\Gamma_0(N)$,
the functions~$I_2$ and~$J_2$ become modular,
\begin{gather*}
\displaystyle I_2=\frac{2 ff''-3f'^2}{f^{4}},
\quad
\displaystyle J_2=\frac{4f^2f'''-24ff'f''+24f'^3}
{f^{6}},
\end{gather*}
and formula (\ref{diff}) reduces to
\[
\frac{\mathrm dI_2}{J_2}=\pi i f(\tau) \mathrm d\tau.
\]
It shows that, up to a constant factor, $f(\tau)\mathrm d\tau$ is a pull-back of the holomorphic differential on~$C$
(see \cite[Section 10.4]{Brezhnev08}, for an alternative derivation of a third-order ODE for modular forms of weight $k=2$).

For several examples of modular forms of weight~$k$ discussed in Section~\ref{sec:ex}, the curve $C$ is a nodal cubic,
\[
F(I_k, J_k)=J_k^2+aI_k^3+bI_k^2=0,
\]
where $a, b \in \mathbb{Q}$ are some constants (that depend on $k$). In view of~(\ref{34}), the corresponding $\mathcal G_k$-invariant fourth-order ODE takes the form
\[
{\cal F}(P_k,Q_k)=2k^2Q_k-2(k+1)P_k+(k+1)a-12(k+2)^2=0,
\]
which is a linear relation between $P_k$ and $Q_k$.
\end{remark}

\begin{remark}\label{F(P,Q)}
For every modular form $f(\tau)$ of weight~$k$, the corresponding ${\cal{G}}_k$-invariant fourth-order ODE (\ref{FPQ}) is necessarily algebraic. In other words, there is another (singular) plane algebraic curve ${\cal C}\colon {{\cal F}(P_k, Q_k)=0}$ associated to every modular form,
furthermore, formulae (\ref{inv}) provide a modular parametrisation of $\cal C$ (note that  $P_k$ and $Q_k$ are modular functions for every~$k$).
There is a natural covering map $C\to {\cal C}$ defined as
\[
P_k=\frac{J_k^2}{I_k^3}, \qquad Q_k=6\frac{(k+2)^2}{k^2}-\frac{k+1}{k^2}\frac{J_kF_{I_k}}{I_k^2F_{J_k}};
\]
use the equation $J_k=\sqrt{P_k}I_k^{\frac32}$ and the second equation (\ref{34}). We emphasise that in all examples discussed so far, the curve $\cal C$ has been rational. Although the rationality of $\cal C$ clearly holds for modular forms $f(\tau)$ on genus zero congruence subgroups $\Gamma_0(N)$
(in which case both $P_k$ and $Q_k$ become rational functions of the corresponding Hauptmodul), it also holds for congruence subgroups of higher genus; see Examples \ref{ExamCong11}--\ref{ExamCong37} of Section \ref{sec:ex} where we derived differential equations for newforms on genus $1$ congruence subgroups $\Gamma_0(11), \Gamma_0(14), \Gamma_0(15)$,  $\Gamma_0(17)$, and $\Gamma_0(37)$.
\end{remark}


\subsection{Examples}
\label{sec:ex}

First of all, we would like to revamp some classical results on the forms~$E_4$, $E_6$, $\Delta$ and~$\theta$.

\newcounter{tbn}\setcounter{tbn}{0}
\noindent\textbf{Example~\refstepcounter{tbn}\thetbn\label{ExamE4E6}: Eisenstein series $E_4$ and $E_6$.}
It has already been mentioned that the Eisenstein series $E_2$, $E_4$ and~$E_6$ satisfy the Ramanujan system~(\ref{Raman}).
Eliminating consequently~$E_4$ and~$E_6$ we arrive at the Chazy equation for~$E_2$,
\[
2E_2'''-2E_2E_2''+3(E_2')^2=0.
\]
Analogously, we can arrive at the equation for~$E_4$, which has appeared earlier  in~\cite[Proposition~4]{Resnikoff1966},
and for~$E_6$, a calculation  which was too `distasteful' for Resnikoff to perform (without the use of  computer), see~\cite[p.~344]{Resnikoff1966}.
The Eisenstein series~$E_4$ satisfies the third-order ODE
\begin{equation}
\begin{aligned}\label{E4}
80E_4^2E_4'''^2+&120(5E_4'^3-6E_4E_4'E_4'')E_4'''+576E_4E_4''^3-20(27E_4'^2+4E_4^3)E_4''^2\\
&+200E_4^2E_4'^2E_4''-125E_4E_4'^4=0,
\end{aligned}
\end{equation}
which takes a nice form
\[
5J_4^2+144I_4^3-80I_4^2=0
\]
in terms of the invariants of the group~$G_4$. In its turn, the $\mathcal G_4$-invariant equation for~$E_4$ is
\[
16Q_4-5P_4-144=0.
\]
The explicit ODEs for~$E_6$ are indeed quite long so that we  only present their invariant forms,
\begin{gather*}
343(J_6^3-216I_6^3)+2(256I_6^3+7J_6^2)^2=0,
\end{gather*}
and
\[
(6Q_6-32)^2-7(Q_6-4)P_6=0.
\]

\noindent\textbf{Example~\refstepcounter{tbn}\thetbn\label{ExamDelta}: Modular discriminant $\Delta$.} The modular discriminant~$\Delta$ is a cusp form of weight~12 defined by the formula
\[
\Delta=q\prod_{n=1}^{\infty}(1-q^n)^{24}.
\]
It can be expressed as $\Delta=\frac1{1728}(E_4^3-E_6^2)$, and thus we can construct a $G_{12}$-invariant third-order ODE satisfied by~$\Delta$, which is
\begin{equation}\label{Delta3}
\begin{array}{c}
36\Delta^4\Delta'''^2-14(18\Delta \Delta''-13\Delta'^2)\Delta^2\Delta'\Delta'''+48\Delta^3\Delta''^3\\
\qquad+285\Delta^2\Delta'^2\Delta''^2-468\Delta\Delta'^4\Delta''+169\Delta'^6+48\Delta^7=0.
\end{array}
\end{equation}
Its invariant form defines an elliptic curve (equianharmonic case),
\[
J_{12}^2+16I_{12}^3+27648=0.
\]
Equation (\ref{Delta3}) was first found in~\cite{Resnikoff1966}.
On the other hand, $\Delta$ is known~\cite{vdP, Rankin}
to satisfy the $\mathcal G_{12}$-invariant fourth-order ODE
\begin{gather}\label{Delta4}
2\Delta^3\Delta''''-10\Delta^2\Delta'\Delta'''-3\Delta^2\Delta''^2+24\Delta\Delta'^2\Delta''-13\Delta'^4=0,
\end{gather}
whose invariant form is
\[
Q_{12}=6.
\]
Note that (\ref{Delta4}) is a differential consequence of (\ref{Delta3}).

\noindent\textbf{Example~\refstepcounter{tbn}\thetbn\label{ExamTheta}: Jacobi theta constants.}
Jacobi theta constants (thetanulls) are defined as
\[
 \theta_2=\sum_{n=-\infty}^{\infty}\mathrm e^{(n-1/2)^2\pi i \tau}, \quad
 \theta_3=\sum_{n=-\infty}^{\infty}\mathrm e^{n^2\pi i \tau}, \quad
 \theta_4=\sum_{n=-\infty}^{\infty}(-1)^n\mathrm e^{n^2\pi i \tau}.
\]
They are known to satisfy the same third-order ODE~\cite{Jacobi1848}
\[
(\theta^2\theta_{\tau\tau\tau}-15\theta\theta_\tau\theta_{\tau\tau}+30\theta_\tau^3)^2+32(\theta\theta_{\tau\tau}-3\theta_\tau^2)^3+\pi^2 \theta^{10}(\theta\theta_{\tau\tau}-3\theta_\tau^2)^2=0.
\]
This equation is $G_{1/2}$-invariant, which reflects the fact that the thetanulls are modular forms of weight~$1/2$,
and can be presented as
\[
J_{1/2}^2+16\, I_{1/2}^3-\frac{1}{16}I_{1/2}^2=0.
\]
Jacobi theta constants also satisfy a nonlinear fourth-order ODE with $\mathrm{GL}(2, \mathbb{R})$-symmetry,
\[
\theta^3(\theta\theta_{\tau\tau}-3\theta_\tau^2)\theta_{\tau\tau\tau\tau}-\theta^4\theta_{\tau\tau\tau}^2+2\theta^2\theta_\tau(\theta\theta_{\tau\tau}+12\theta_\tau^2)\theta_{\tau\tau\tau}+\theta^3\theta_{\tau\tau}^3
-24\theta^2\theta_\tau^2\theta_{\tau\tau}^2-18\theta\theta_\tau^4\theta_{\tau\tau}+18\theta_\tau^6=0,
\]
see \cite[eq.~(5.5)]{O3}, whose invariant form is
\[
Q_{1/2}-6P_{1/2}-102=0.
\]

Now we would like to consider some recently discovered systems of ODEs for modular forms on congruence subgroups of~$\mathrm{SL}(2,\mathbb Z)$.
Thus, there are Ramanujan-like systems for modular forms of level two~\cite{Ramamani1917}, three~\cite{Huber2011,Matsuda2020,O1}, five~\cite{Mano2002,Matsuda2019}, six~\cite{Matsuda2022}, etc.
(See many more results on various relations between modular forms on congruence subgroups of~$\mathrm{SL}(2,\mathbb Z)$ in~\cite{Cooper2017}.)

\noindent\textbf{Example~\refstepcounter{tbn}\thetbn\label{ExamCongTwo}: Modular forms on~$\Gamma_0(2)$.} Ramamani~\cite{Ramamani1917} found the analogue
\begin{gather}
\mathcal P'=\frac{\mathcal P^2-\mathcal Q}4,\quad
\tilde{\mathcal P}'=\frac{\mathcal P\tilde{\mathcal P}-\mathcal Q}2, \quad
\mathcal Q'=(\mathcal P-\tilde{\mathcal P})\mathcal Q,
\end{gather}
of the Ramanujan system for the congruence subgroup~$\Gamma_0(2)$ of~$\mathrm{SL}(2,\mathbb Z)$,
\[
\Gamma_0(2):=\left\{\begin{pmatrix} a & b\\ c & d\end{pmatrix}\in\mathrm{SL}(2,\mathbb Z)\mid c\equiv0\bmod 2\right\}.
\]
Here, $\mathcal P$ and~$\mathcal Q$ are normalized Eisenstein series on~$\Gamma_0(2)$ of weights~2 and~4, respectively,
and $\tilde{\mathcal P}$ is a modular form of weight~2,
\[
\mathcal P(q)=1-8\sum\limits_{n=1}^\infty\frac{(-1)^nnq^n}{1-q^n},\quad
\mathcal Q(q)=1+16\sum\limits_{n=1}^\infty\frac{(-1)^nn^3q^n}{1-q^n},\quad
\tilde{\mathcal P}(q)=1+24\sum\limits_{n=1}^\infty\frac{nq^n}{1+q^n}.
\]
It was shown in~\cite{Ablowitz2006} that~$\tilde{\mathcal P}$ can be presented as $\tilde{\mathcal P}=\frac32\mathcal P-\frac12E_2$ where
$\mathcal P$ satisfies a third-order ODE,
\[
2(\mathcal P^2-4\mathcal P')\mathcal P'''+8(\mathcal P'')^2-2\mathcal P^3\mathcal P''+(3\mathcal P^2-4\mathcal P')(\mathcal P')^2=0,
\]
and the weight-4 modular form $\mathcal D=\frac1{64}(\tilde{\mathcal P}^2-\mathcal Q)$, which can be seen as an analogue of~$\Delta$ for the congruence group~$\Gamma_0(2)$, satisfies a fourth-order ODE,
\begin{gather*}
(8\mathcal D^4\mathcal D''-10\mathcal D^3\mathcal D'^2)\mathcal D''''-8\mathcal D^4\mathcal D'''^2
+(10\mathcal D^2\mathcal D'^3+16\mathcal D^3\mathcal D'\mathcal D'')\mathcal D'''\\
-20\mathcal D^3\mathcal D''^3+39\mathcal D^2\mathcal D'^2\mathcal D''^2
-60\mathcal D\mathcal D'^4\mathcal D''+25\mathcal D'^6=0.
\end{gather*}
We can build on these results by showing additionally that $\tilde{\mathcal P}$, which is a modular form of weight two, satisfies the third-order ODE
\begin{gather*}
3\tilde{\mathcal P}^2\tilde{\mathcal P}'''^2-36(\tilde{\mathcal P}\tilde{\mathcal P}'\tilde{\mathcal P}''-\tilde{\mathcal P}'^3)\tilde{\mathcal P}'''
+32\tilde{\mathcal P}\tilde{\mathcal P}''^3-3(12\tilde{\mathcal P}'^2+\tilde{\mathcal P}^4)\tilde{\mathcal P}''^2
+9\tilde{\mathcal P}^3\tilde{\mathcal P}'^2\tilde{\mathcal P}''-\frac{27}{4}\tilde{\mathcal P}^2\tilde{\mathcal P}'^4=0,\\
\qquad \text{or, in invariant terms,} \quad 3J_2^2+64I_2^3-12I_2^2=0,
\end{gather*}
the weight-4 modular form~$\mathcal Q$ satisfies the third-order ODE
\begin{gather*}
4\mathcal Q^4\mathcal Q'''^2-6\mathcal Q^2\mathcal Q'(6\mathcal Q\mathcal Q''-5\mathcal Q'^2)\mathcal Q'''+16\mathcal Q^3\mathcal Q''^3+
\mathcal Q^2(21\mathcal Q'^2-4\mathcal Q^3)\mathcal Q''^2\\
-10\mathcal Q\mathcal Q'^2(6\mathcal Q'^2-\mathcal Q^3)\mathcal Q''+25\mathcal Q'^6-\frac{25}4\mathcal Q^3\mathcal Q'^4=0,\\
\qquad \text{or, in invariant terms,} \quad J_4^2+16I_4^3-16I_4^2=0,
\end{gather*}
the weight-4 modular form~$\mathcal D$ satisfies the third-order ODE
\begin{gather*}
4\mathcal D^4\mathcal D'''^2-6\mathcal D^2\mathcal D'(6\mathcal D\mathcal D''-5\mathcal D'^2)\mathcal D'''+16\mathcal D^3\mathcal D''^3-
\mathcal D^2(256\mathcal D^3-21\mathcal D'^2)\mathcal D''^2\\
+20\mathcal D\mathcal D'^2(32\mathcal D^3-3\mathcal D'^2)\mathcal D''+25\mathcal D'^6
-400\mathcal D^3\mathcal D'^4=0,\\
\qquad \text{or, in invariant terms,} \quad J_4^2+16I_4^3-1024I_4^2=0,
\end{gather*}
and the weight-8 cusp form $\tilde{\mathcal D}=\mathcal D\mathcal Q$ satisfies the third-order ODE
\begin{gather*}
16\tilde{\mathcal D}^4\tilde{\mathcal D}'''^2-30\tilde{\mathcal D}^2\tilde{\mathcal D}'(4\tilde{\mathcal D}\tilde{\mathcal D}''-3\tilde{\mathcal D}'^2)\tilde{\mathcal D}'''+32\tilde{\mathcal D}^3\tilde{\mathcal D}''^3+117\tilde{\mathcal D}^2\tilde{\mathcal D}'^2\tilde{\mathcal D}''^2-8\tilde{\mathcal D}(16\tilde{\mathcal D}^5+27\tilde{\mathcal D}'^4)\tilde{\mathcal D}''\\
+81\tilde{\mathcal D}'^6+144\tilde{\mathcal D}^5\tilde{\mathcal D}'^2=0,\\
\qquad \text{or, in invariant terms,} \quad J_8^2+16I_8^3-496I_8=0.
\end{gather*}
Note that the last invariant equation defines an elliptic curve (lemniscatic case).
A $\mathcal G_8$-invariant ODE for~$\tilde{\mathcal D}$ is $128Q_8-9P_8-912=0$.
We refer to Remark~\ref{SpecificCubics} of Section \ref{sec:invar} for the fourth-order $\mathcal G_k$-invariant equations for the modular forms $\tilde{\mathcal P}, \mathcal Q, \mathcal D$
(note that $\mathcal P$ is only quasi-modular).

\noindent\textbf{Example~\refstepcounter{tbn}\thetbn\label{ExamE13}: Eisenstein series $E_{1,3}$.} This modular form of weight~$1$ and level~$3$ is defined as
\begin{equation*}
E_{1, 3}(\tau)=\sum_{(\alpha,\beta)\in\Z^2}q^{(\alpha^2-\alpha\beta+\beta^2)}=1+6q+6q^3+6q^4+12q^7+.....
\label{g1}
\end{equation*}
Matsuda \cite{Matsuda2020} presented several systems satisfied by~$E_{1,3}$, one of which was first derived by Huber in~\cite{Huber2011},
\begin{equation*}
E_{1,3}'=\frac{M E_{1,3}-N}3,\quad M'=\frac{M^2-N E_2}3,\quad N'=(M-E_2^2)N,
\end{equation*}
where~$M$ and~$N$ are some functions whose specific forms are irrelevant here.
Eliminating $M$ and $N$, we obtain a $G_1$-invariant third-order equation for~$f=E_{1,3}$,
\begin{gather}\label{E_13Ord3}
f^2f'''^2-6f'(3ff''-4f'^2)f'''+18ff''^3 -(f^6+27f'^2)f''^2 +4f^5f'^2f''-4f^4f'^4=0,
\end{gather}
whose invariant form is
\[
J_1^2+18I_1^3-I_1^2=0.
\]
It was also shown in \cite{FO} that $f=E_{1,3}$ satisfies the fourth-order ODE
\begin{gather}\label{E_13Ord4}
ff''''(ff''-2f'^2)-f^2f'''^2+2f'(ff''+4f'^2)f'''-9f'^2f''^2=0,
\end{gather}
whose invariant form is
\[
Q_1-2P_1-36=0.
\]

Finally, if one does not know a particular ODE or a system thereof satisfied by a given modular form,
it is possible to construct a third-order ODE directly by choosing an appropriate (polynomial in~$I_k$ and~$J_k$) ansatz. This is what we have done in the examples below.

\noindent\textbf{Example~\refstepcounter{tbn}\thetbn\label{ExamCong11}: Newform of weight~2 on $\Gamma_0(11)$.}
There is a unique cusp form of weight~$2$ on the congruence subgroup~$\Gamma_0(11)$,
\[
f(\tau)=q\prod\limits_{n=1}^\infty (1-q^n)^2(1-q^{11n})^2,
\]
labelled as case 11.2.a.a in the LMFDB database~\cite{LMFDB}.
It satisfies a $G_{11}$-invariant equation
\begin{gather*}
J_2^4+32(I_2-8)(I_2^2+72I_2-944)J_2^2+256(I_2+8)(I_2^3+152I_2^2-704I_2+1168)(I_2-8)^2=0,
\end{gather*}
which defines a singular algebraic curve $C$ of genus~$1$ with the $j$-invariant $j=-2^{12}31^311^{-5}$.
Remarkably, this value coincides with the $j$-invariant of the curve $y^2+y=x^3-x^2-10x-20$,
which can be uniformised by the same cusp form~$f$, see~\cite[Table 1, case $11B$]{BK}.
Thus, both curves are birationally equivalent via the map
\[
I_2=-\frac{8x^2+8x-119}{(x-5)^2},\quad J_2=-\frac{44(4x-9)(2y+1)}{(x-5)^3}.
\]
The inverse transformation is
\begin{gather*}
x=\frac{-11J_2^2+16(I_2-8)(9I_2^2-2468I_2+11712)}{64(I_2-83)(I_2^2-64)},\\
y=\frac{(11I_2-264)J_2^3+176(I_2-8)(I_2^3+80I_2^2-5584I_2+43904)J_2}{512(I_2^2-64)^2(I_2-83)}+\frac12.
\end{gather*}
Note that ${\mathrm dI_2}/{J_2}$ is the holomorphic differential on $C$.
By Remark~\ref{SpecificCubics} of Section~\ref{sec:invar}, it equals $\pi i f(\tau) \mathrm d\tau$.
The form $f(\tau)$ also satisfies the fourth-order $\mathcal G_2$-invariant ODE
\begin{gather*}
5616022359375P_2^4-2^43^55^3(34618195Q_2-763426383)P_2^3\\
-2^83^3(173368000Q_2^3-8479136175Q_2^2+183916606320Q_2-1561600055241)P_2^2\\
+2^{12}3^2(64349800Q_2^4{-}3828348951Q_2^3{+}88775864253Q_2^2{-}1000262056761Q_2{+}4759648412715)P_2\\
+131072(4Q_2-63)(5329Q_2^3-204861Q_2^2+2745099Q_2-14039703)(2Q_2-105)^2=0,
\end{gather*}
which defines a singular rational curve~$\cal C$ in the plane $(P_2,Q_2)$.

\noindent\textbf{Example~\refstepcounter{tbn}\thetbn\label{ExamCong14}: Newform of weight~2 on $\Gamma_0(14)$.} A unique cusp form of weight~$2$ on the congruence subgroup~$\Gamma_0(14)$,
labelled as 14.2.a.a in~\cite{LMFDB},
\[
f(\tau)=q\prod\limits_{n=1}^\infty (1-q^n)(1-q^{2n})(1-q^{7n})(1-q^{14n}),
\]
satisfies the following third-order ODE,
\begin{gather*}
J_2^4+32(I_2-5)(I_2^2+49I_2-350)J_2^2+256(I_2+20)(I_2-4)(I_2^2+86I_2-199)(I_2-5)^2=0.
\end{gather*}
The genus of this singular algebraic curve $C$ is~$1$, the $j$-invariant is $j=5^3 11^3 31^3 2^{-3}7^{-6}$ and the holomorphic differential is ${\mathrm dI_2}/{J_2}$.
This curve is birationally equivalent to the curve~\cite[Table~1, case~14D]{BK}, $y^2+(x+1)y=x^3-36x-70$, via the transformation
\[
I_2=\frac{4x^2+18x-41}{(x+4)^2},\quad J_2=-\frac{28(x+11)(2y+x+1)}{(x+4)^3}.
\]
The curve $\cal C$ corresponding to the fourth-order $\mathcal G_2$-invariant ODE for $f(\tau)$ is rational (we do not present it explicitly due to its complexity).

\noindent\textbf{Example~\refstepcounter{tbn}\thetbn\label{ExamCong15}: Newform of weight~2 on $\Gamma_0(15)$.} A unique cusp form of weight~$2$ on the congruence subgroup~$\Gamma_0(15)$,
labelled as 15.2.a.a in~\cite{LMFDB},
\[
f(\tau)=q\prod\limits_{n=1}^\infty (1-q^n)(1-q^{3n})(1-q^{5n})(1-q^{15n}),
\]
satisfies the following third-order ODE,
\begin{gather*}
J_2^4+32(I_2-5)(I_2^2+33I_2-406)J_2^2+256(I_2+76)(I_2+4)(I_2^2-10I_2+89)(I_2-5)^2=0.
\end{gather*}
The genus of this singular algebraic curve~$C$ is~$1$, the $j$-invariant is $j=23^373^33^{-2}5^{-8}$ and the holomorphic differential is ${\mathrm dI_2}/{J_2}$.
This curve is birationally equivalent to the curve~\cite[Table~1, case~15F]{BK}, $y^2+(x+1)y=x^3+x^2+35x-28$, via the transformation
\[
I_2=-\frac{4x^2+74x+61}{(x-2)^2},\quad J_2=-\frac{180(x+3)(2y+x+1)}{(x-2)^3}.
\]
The curve $\cal C$ corresponding to the fourth-order $\mathcal G_2$-invariant ODE for $f(\tau)$ is rational (we do not present it explicitly due to its complexity).

\noindent\textbf{Example~\refstepcounter{tbn}\thetbn\label{ExamCong17}: Newform of weight~2 on $\Gamma_0(17)$.}
A cusp form  of weight~$2$ on the congruence subgroup~$\Gamma_0(17)$ labelled as 17.2.a.a in \cite{LMFDB},
\[
f(\tau)=q-q^{2}-q^{4}-2q^{5}+4q^{7}+3q^{8}-3q^{9}+O(q^{10}),
\]
satisfies the  third-order ODE
\begin{gather*}
9J_2^8{+}96(7I_2^3{+}3I_2^2{-}1734I_2{+}21340)J_2^6\\
{+}256(73I_2^6{+}156I_2^5{-}45153I_2^4{+}734762I_2^3{-}2275710I_2^2{-}12691872I_2{+}175400560)J_2^4\\
{+}8192(28I_2^5{+}993I_2^4{-}11044I_2^3{+}13213I_2^2{+}854526I_2{+}6877196)(I_2^4{-}30I_2^3{+}309I_2^2{-}584I_2{+}5232)J_2^2\\
{+}65536(16I_2^4{+}1132I_2^3{+}13477I_2^2{+}91338I_2{+}212581)(I_2^4{-}30I_2^3{+}309I_2^2{-}584I_2{+}5232)^2=0,
\end{gather*}
which defines a singular algebraic curve $C$ of genus~$1$ with the $j$-invariant $j=-3^3 11^3 17^{-4}$. Once again, this value coincides with the $j$-invariant of the curve $y^2+(x+1)y=x^3-x^2-x-14$, which can be uniformised by the same cusp form, see~\cite[Table 1, case 17C]{BK}. Thus, both curves are birationally equivalent.
The holomorphic differential is again ${\mathrm dI_2}/{J_2}$.

The curve $\cal C$ corresponding to the fourth-order $\mathcal G_2$-invariant ODE for $f(\tau)$ is rational (we do not present it explicitly due to its complexity).

\noindent\textbf{Example~\refstepcounter{tbn}\thetbn\label{ExamCong37}: Newform of weight~2 on $\Gamma_0(37)$.}
A cusp form  of weight~$2$ on the congruence subgroup~$\Gamma_0(37)$ labelled as 37.2.a.b in \cite{LMFDB},
\[
f(\tau)=q+q^{3}-2q^{4}-q^{7}-2q^{9}+O(q^{10}),
\]
satisfies a  third-order ODE $F(I_{37}, J_{37})=0$
which defines a singular algebraic curve $C$ of genus~$1$ with the $j$-invariant $j=2^{12}3^3 37^{-1}$. This value coincides with the $j$-invariant of the curve $y^2-y=x^3-x$, which can be uniformised by the same cusp form, \cite[Table 1, case 37A]{BK}, compare with \cite{Zagier1984}. Thus, both curves are birationally equivalent.
The holomorphic differential is ${\mathrm dI_2}/{J_2}$. Note that in this example the modular curve $\Gamma_0(37)\backslash {\mathcal H}$ has genus two and, according to  \cite{Zagier1985}, the map $\Gamma_0(37)\backslash {\mathcal H}\rightarrow C$ is a two-sheeted covering.

The curve $\cal C$ corresponding to the fourth-order $\mathcal G_2$-invariant ODE for $f(\tau)$ is rational (we do not present it explicitly due to its complexity).

\section{Differential equations for Jacobi forms}
\label{sec:Jac}

\begin{Defi}
A Jacobi form of weight~$k$ and index~$m$ is a holomorphic function $f\colon \mathcal H\times \mathbb C\mapsto \mathbb C$ with the transformation property
\[
\tilde\tau = \frac{a\tau+b}{c\tau+d},\quad \tilde z = \frac{z+\lambda \tau + \mu}{c\tau+d}, \quad
\tilde f = (c\tau+d)^k\mathrm e^{2\pi im\left(\frac{c(z+\lambda\tau+\mu)^2}{c\tau+d}-\lambda^2\tau-2\lambda z-\lambda\mu\right)}f,
\]
where $\tau\in \mathcal H$, $z\in\mathbb C$, $\begin{pmatrix} a & b \\ c & d \end{pmatrix}\in \mathrm{SL}(2,\mathbb Z)$, $\lambda, \mu\in\mathbb Z$,
and such that it has a Fourier expansion of the form
\[
f(\tau,z)=\sum_{\substack{n=0\\r\in\mathbb Z,\ r^2\leqslant 4nm}}^\infty c(n,r) q^n\zeta^r,
\]
where $q=\mathrm e^{2\pi i\tau}$ and $\zeta=\mathrm e^{2\pi iz}$.
If $f$ has a Fourier expansion of the same form but with $r^2<4nm$ then $f$ is called a Jacobi cusp form of weight~$k$ and index~$m$.
If we drop the holomorphicity condition, the function~$f$ is called a weak Jacobi form if $c(n,r)=0$ unless $n\geqslant n_0$ for some possible negative integer~$n_0$
\end{Defi}

Classical examples of Jacobi forms are modular forms, theta series, Fourier coefficients of Siegel modular forms and the Weierstrass $\wp$-function.
See more details in~\cite{EichlerZagier1985}.

Similarly to modular forms, there is a notion of the Rankin--Cohen bracket~\cite{Choie1997} that assigns to Jacobi forms~$f_1$ and~$f_2$ of
weights~$k_1$ and~$k_2$ and indices~$m_1$ and~$m_2$, respectively, a Jacobi form of weight~$k_1+k_2+2n$ and index~$m_1+m_2$,
\[
[[f_1,f_2]]_n:=\sum\limits_{i=0}^n(-1)^i{k_1+n-3/2 \choose n-i}{k_2+n-3/2 \choose i}m_1^{n-i}m_2^i L_{m_1}^i(f_1) L_{m_2}^{n-i}(f_2),
\]
where $L_m:=8\pi i m\, \partial_{\tau}-\partial_z^2$ is the heat operator.
The theory of Jacobi forms features an additional family of Rankin--Cohen operators~\cite{ChoieEholzer1998} parametrised by an arbitrary complex number~$X$
and a non-negative integer~$n$,
\begin{gather*}
[f_1,f_2]_{X,2n}:=\sum\limits_{r+s+p=n} C_{r,s,p}(k_1,k_2)(1+m_1X)^s(1-m_2X)^r
L_{m_1+m_2}^p(L_{m_1}^r(f_1)L_{m_2}^s(f_2)),\\
[f_1,f_2]_{X,2n+1}=m_1[f_1,\p_zf_2]_{X,2n}-m_2[\p_zf_1,f_2]_{X,2n}
\end{gather*}
where
\[
C_{r,s,p}(k_1,k_2):=\frac{(k_1+n-3/2)_{s+p}}{r!}\frac{(k_2+n-3/2)_{r+p}}{s!}\frac{(3/2-k_1-k_2-n)_{r+s}}{p!}
\]
and $(x)_l=\prod\limits_{0\leqslant i\leqslant l-1}(x-i)$.

\subsection{Differential invariants of the Jacobi group}
\label{sec:Jac0}
Consider the six-dimensional group~$G_{k,m}$ of point transformations called the Jacobi group and acting on a space with coordinates~$(\tau,z,f)$,
\[
\tilde\tau = \frac{a\tau+b}{c\tau+d},\quad \tilde z = \frac{z+\lambda \tau + \mu}{c\tau+d}\quad
\tilde f = (c\tau+d)^k\mathrm e^{2\pi im\left(\frac{c(z+\lambda\tau+\mu)^2}{c\tau+d}-\lambda^2\tau-2\lambda z-\lambda\mu+\kappa\right)}f,
\]
where $\begin{pmatrix} a & b \\ c & d \end{pmatrix}\in \mathrm{SL}(2,\mathbb R)$ and $\lambda, \mu, \kappa \in \mathbb{R}$.
In the \textit{generic} case
\[
2kff_{zz}{-}8m\pi iff_\tau{-}(2k{-}1)f_z^2\neq0\quad \text{and}\quad  m\neq0,
\]
its second-order differential invariants are
\begin{align*}
L_{k,m} =& \frac1{(2kff_{zz}{-}8m\pi iff_\tau{-}(2k{-}1)f_z^2)^2}
\Big(64m^2\pi ^2f^3f_{\tau\tau}{+}32m\pi i f^2f_zf_{\tau z}{-}4k(k{+}1)f^2f_{zz}^2\\
&{+}4(8m(k{+}1)\pi iff_\tau{+}(2k^2{+}k{-}2)f_z^2)ff_{zz}{-}(16(2k{+}3)m\pi iff_\tau{+}(4k^2{-}7)f_z^2)f_z^2\Big),\\
M_{k,m} =& \frac{4m\pi if^2f_{\tau z}{-}ff_zf_{zz}{-}(4m\pi iff_\tau{-}f_z^2)f_z}{(2kff_{zz}{-}8m\pi iff_\tau{-}(2k{-}1)f_z^2)^{3/2}},
\end{align*}
and the third-order differential invariants are
\begin{align*}
N_{k,m} =& \frac1{(2kff_{zz}{-}8m\pi iff_\tau {-}(2k{-}1)f_z^2)^3}
\Big(512m^3\pi^3if^5f_{\tau \tau \tau }{-}384m^2\pi^2f^4f_zf_{\tau \tau z}{-}96m\pi if^3f_z^2f_{\tau zz}\\
&{+}8f^2f_z^3f_{zzz}{-}96m \pi f^2(2(k{+}2)ff_{zz}{-}(2k{+}5)f_z^2)(2m\pi ff_{\tau \tau }{+}if_zf_{\tau z})\\
&{-}4(k{+}2)f^2f_{zz}^2\left(24(k+1)m\pi iff_\tau {-}2k(k+1)ff_{zz}{+}3(k-1)(2k{+}3)f_z^2\right)\\
&{-}6(16(2k{+}3)(k{+}2)m\pi iff_\tau {+}(4k^3{+}8k^2{-}11k{-}24)f_z^2)ff_z^2f_{zz}\\
&{-}24(4k^2{+}16k{+}17)m\pi iff_\tau f_z^4{-}(8k^3{+}12k^2{-}38k{-}65)f_z^6\Big),
\end{align*}
\begin{align*}
P_{k,m}=&\frac1{(2kff_{zz}{-}8m\pi iff_\tau {-}(2k{-}1)f_z^2)^{5/2}}
\Big(16m^2\pi^2f^4f_{\tau \tau z}{+}8m\pi if^3f_zf_{\tau zz}{-}f^2f_z^2f_{zzz}\\
&{-}16m^2\pi^2f^3f_zf_{\tau \tau }{+}4(2(k{+}2)m\pi iff_{zz}{-}(2k{+}7)m\pi if_z^2)f^2f_{\tau z}{-}2(k{+}2)f^2f_zf_{zz}^2\\
&{-}2((4 m\pi i(k+2)ff_\tau {-}(2k+5)f_z^2)ff_zf_{zz}{+}(4m\pi iff_\tau {-}f_z^2)(2k{+}5)f_z^3\Big),\\
Q_{k,m}=&\frac{4m\pi i(ff_{\tau zz}{-}2f_zf_{\tau z})f^2{-}f^2(f_zf_{zzz}{+}f_{zz}^2){-}(4m\pi iff_\tau {-}5f_z^2)ff_{zz}{+}(8m\pi iff_\tau {-}3f_z^2)f_z^2}
{(2kff_{zz}{-}8m\pi iff_\tau {-}(2k{-}1)f_z^2)^2},\\
R_{k,m}=&\frac{f^2f_{zzz}{-}3ff_zf_{zz}{+}2f_z^3}{(2kff_{zz}{-}8m\pi iff_\tau {-}(2k{-}1)f_z^2)^{3/2}}.
\end{align*}
The Lie algebra of the group~$G_{k,m}$ is spanned by the vector fields
\[
\langle \p_\tau,\ 2\tau\p_\tau+z\p_z,\ \tau^2\p_\tau+\tau z\p_z-(2m\pi iz^2f+k\tau f)\p_f,\ \tau\p_z-4m\pi izf\p_f,\ \p_z,\ f\p_f\rangle.
\]
Note that the invariance of the above expressions under $G_{k,m}$-action makes them Jacobi functions, that is, Jacobi forms or weight~0 and order~0.
Furthermore, their numerators and denominators are Jacobi forms of the same weight and index.
Therefore, there is an analogy between invariants of the group~$G_{k,m}$ and invariants of the groups~$G_k$ and~$\mathcal G_k$.
The expression $D=2kff_{zz}-8m\pi iff_\tau-(2k-1)f_z^2$ is a conditional invariant of the group~$G_{k,m}$
(that is, its vanishing is an invariant condition). Moreover, given a Jacobi form~$f$ of weight~$k$ and index~$m$,
the function~$D$ is a Jacobi form of weight~$2k+2$ and index~$2m$.
Indeed, $mfD=[[f,f^2]]_1$, or alternatively $D=\frac{2}{2k-1}[f,f]_{X,2}$, with the right-hand side being actually $X$-independent.
Introducing
\begin{gather*}
M:=\frac{[f,f^2]_{\frac{2k}{m(10k-3)},3}}{6k-1},\quad
Q:=\frac{(10k-3)[f,M]_{X,1}}{3m^2(2k-1)(4k-1)},\\
L:=\frac{\left[f^2,\left[f,f\right]_{X,2}\right]_{-\frac{2k}{m(8k+3)},2}+k(32k^3-14k+3)Q
-\frac{(4k+3)(4k+11)}{(2k-1)}[f,f]_{X,2}^2}
{4(2k-1)(4k-1)(4k+3)},
\end{gather*}
we can write down all the above invariants of the $G_{k,m}$-action in terms of Rankin--Cohen brackets,
\begin{gather*}
L_{k,m}=\frac{(2k-1)^2L}{4[f,f]_{X,2}^2},\quad
M_{k,m}=-\frac{(10k-3)M}{3m(2k-1)(4k-1)(6k-1)\left(\frac{2}{2k-1}[f,f]_{X,2}\right)^{3/2}},
\end{gather*}
\begin{gather*}
N_{k,m}=
\frac{\frac{1}{4(8k-7)}[f^2,L]_{\frac{-2k}{m(12k+7),2}}-\frac{2(4k^2+3k-1)}{(2k-1)}[f,f]_{X,2}Q-\frac{4k+23}{4(2k-1)}[f,f]_{X,2}L-\frac{2(10k-3)^2}{9m^2(2k-1)^2(4k-1)}M^2}{(4k-1)\frac{8}{(2k-1)^3}[f,f]_{X,2}^3}\\
P_{k,m}=\frac{
\frac{8(k+2)(10k-3)}{(2k-1)(4k-1)}\left[f,f^2\right]_{\frac{2k}{m(10k-3)},3}[f,f]_{X,2}
-\frac{10k-3}{(4k-1)(6k+5)}\left[\left[f,f^2\right]_{\frac{2k}{m(10k-3)},3},f^2\right]_{\frac{2k}{5m(2k+1)},2}}{3m(2k-1)(4k-1)(6k-1)(\frac{2}{2k-1}[f,f]_{X,2})^{5/2}},\\
Q_{k,m}=\frac{-(2k-1)^2Q}{4[f,f]_{X,2}^2},\quad
R_{k,m}=\frac{8(3k-1)[f,f^2]_{\frac1{4m(3k-1)},3}}{3(2k-1)(4k-1)(6k-1)(\frac{2}{2k-1}[f,f]_{X,2})^{3/2}}.
\end{gather*}
The expressions for the brackets with the unspecified parameter~$X$ do not involve~$X$ in their expanded forms.

Any {\it generic} third-order $G_{k,m}$-invariant involutive PDE system (governing Jacobi forms) can be obtained by expressing all third-order invariants $N_{k,m}$, $P_{k,m}$, $Q_{k,m}$, $R_{k,m}$ as functions of the second-order invariants $L_{k,m}, M_{k,m}$, see Examples~\ref{Jacobi-12} and~\ref{Jacobi-21} where we present such systems for the weak Jacobi forms $\varphi_{-1, 2}(\tau, z)$ and  $\varphi_{-2, 1}(\tau, z)$. The action of the group $G_{k,m}$ on the six-dimensional solution space of any such system is locally transitive (possesses an open orbit).

There also exist two different {\it non-generic} $G_{k,m}$-invariant involutive PDE systems, first of which contain the equation $D=0$, equivalently,
\begin{subequations}\label{fSys}
\begin{gather}\label{fSysHeat}
8m\pi i f_\tau=2kf_{zz}-(2k-1)f_z^2/f;
\end{gather}
note that the substitution $f=\mathrm e^{\varphi}$ reduces~(\ref{fSysHeat}) to a potential Burgers equation for $\varphi$, namely,
$8m\pi i \varphi_\tau=2k\varphi_{zz}+\varphi_z^2$, while the substitution $f=\psi^{2k}$ linearises equation  (\ref{fSysHeat}) to $4m\pi i \psi_t=k\psi_{zz}$.
In the non-generic case the above invariants are not defined,
and therefore we consider instead
\begin{gather*}
\mathcal L_{k,m}=\frac1{L_{k,m}},\quad
\mathcal M_{k,m}=\frac{L_{k,m}^{\frac12}}{M_{k,m}^{\frac23}},\quad
\mathcal N_{k,m}=\frac{N_{k,m}}{L_{k,m}^{\frac32}},\\
\mathcal P_{k,m}=\frac{P_{k,m}}{L_{k,m}^{\frac12}M_{k,m}},\quad
\mathcal Q_{k,m}=\frac{Q_{k,m}}{L_{k,m}},\quad
\mathcal R_{k,m}=\frac{R_{k,m}}{M_{k,m}}.
\end{gather*}
In view of equation~(\ref{fSysHeat}), the invariants $\mathcal L_{k,m}$, $\mathcal Q_{k,m}$ and $\mathcal R_{k,m}$ have fixed values,
$\mathcal L_{k,m}=0$, $\mathcal Q_{k,m}=-\frac1{4k}$, $\mathcal R_{k,m}=\frac1k$, and thus the essential invariants are~$\mathcal M_{k,m}$,
$\mathcal N_{k,m}$ and~$\mathcal P_{k,m}$, which involve $z$-derivatives of~$f$ only and are of order~4, 6 and~5, respectively.
In particular, to obtain a non-generic $G_{k,m}$-invariant involutive PDE system with a transitive $G_{k,m}$-action on the solution space,
one has to add to equation~(\ref{fSysHeat}) a sixth-order equation that is a function of~$\mathcal M_{k,m}$, $\mathcal N_{k,m}$ and~$\mathcal P_{k,m}$.
One particular choice of this kind is
\begin{gather}\label{fSysSix}
\left(\frac{(\ln f)_{zzzzz}}{(\ln f)_{zzz}}+\frac{6}{k}(\ln f)_{zz}\right)_z=0.
\end{gather}
\end{subequations}
The system~\eqref{fSys} is in involution and admits $G_{k,m}$ as a symmetry group.
It can be represented in invariant form as
\[
\mathcal L_{k,m}=0,\quad \mathcal M_{k,m}^3(\mathcal N_{k,m}-2\mathcal P_{k,m})-8=0,
\]
see Example~\ref{JacobiTheta} of Section~\ref{sec:Jac1} where we obtain a system of this kind for the Jacobi theta functions.

Another type of {\it non-generic} $G_{k,m}$-invariant involutive PDE systems is associated with the value $m=0$.
In this case, the invariants $L_{k,0}$ and~$M_{k,0}$ are functionally dependent.
Moreover, the group~$G_{k,0}$ is five-dimensional since the $f$-scalings are no longer admissible.
The Lie algebra of the group~$G_{k,0}$ is spanned by the vector fields
\[
\langle \p_\tau,\ \p_z,\ \tau\p_z,\ 2\tau\p_\tau+z\p_z-kf\p_f,\ \tau^2\p_\tau+\tau z\p_z-k\tau f\p_f\rangle.
\]
The invariants of the group~$G_{k,0}$ that are necessary for the exposition below are as follows
\begin{gather*}
\mathscr I_k=\frac{(kff_{zz}-(k+1)f_z^2)f_{\tau\tau}-kff_{\tau z}^2+(k+1)f_\tau(2f_zf_{\tau z}-f_\tau f_{zz})}{f^{1+\frac4k}(kff_{zz}-(k+1)f_z^2)},\\
\mathscr J_k= \frac{(kff_{zz}-(k+1)f_z^2)f_{\tau zz}-(kff_{\tau z}-(k+1)f_\tau f_z^2)f_{zzz}+(k+2)f_{zz}(f_zf_{\tau z}-f_\tau f_{zz})}
{f^{1+\frac4k}(kff_{zz}-(k+1)f_z^2)},\\
\mathscr L_k=\frac{f_z}{f^{1+\frac1k}},\quad
\mathscr M_k=\frac{f_{zz}}{f^{1+\frac2k}},\quad
\mathscr N_k= \frac{f_{zzz}}{f^{1+\frac3k}},
\end{gather*}
see Example~\ref{Weierstrass} of Section~\ref{sec:Jac1} where derive differential system for the Weierstrass $\wp$-function.

\subsection{Examples}
\label{sec:Jac1}
Here we provide examples of nonlinear involutive third-order PDE systems that characterise Jacobi forms uniquely up to the action of the corresponding six-dimensional symmetry group~$G_{k,m}$. We refer to~\cite{AG} for an alternative construction of linear modular differential equations satisfied by Jacobi forms.

\noindent {\bf Example~\refstepcounter{tbn}\thetbn\label{Jacobi-12}.} The weak Jacobi form $\varphi_{-1, 2}(\tau, z)$ of weight~$-1$ and index~$2$ is defined as
$\varphi_{-1, 2}(\tau, z)=\Delta^{-1/8}(\tau)\vartheta_1(\tau, 2z)$ where~$\Delta$ is the modular discriminant and $\vartheta_1$ is a Jacobi theta function, see~\cite[formula~(4.31)]{DMZ},
\[
\vartheta_1(\tau,z)=2\sum_{n=0}^\infty (-1)^{n+1}\mathrm e^{\pi i\left(n+\frac12\right)\tau}\sin((2n+1)\pi z).
\]
The function $f(\tau, z)=\varphi_{-1, 2}(\tau, z)$ satisfies the following $G_{-1, 2}$-invariant overdetermined involutive PDE system,
\begin{gather*}
2\pi iff_{\tau \tau \tau} = 2\pi if_\tau f_{\tau \tau}+f_{\tau \tau}f_{zz}-f_{\tau z}^2,\\
ff_{\tau \tau z} = 3f_zf_{\tau \tau}-2f_\tau f_{\tau z},\\
ff_{\tau zz} = 8\pi i(ff_{\tau \tau}-f_\tau^2)+2f_zf_{\tau z}-f_\tau f_{zz},\\
ff_{zzz} = 16\pi i(ff_{\tau z}-f_\tau f_z)+f_zf_{zz},
\end{gather*}
which can be obtained from the system~(\ref{eqg}) by the change of variables
\[
f(\tau,z) = g(\tilde \tau,\tilde z), \quad \tau = \frac{i\tilde\tau}\pi, \quad z = \frac{\tilde z}{2\pi}.
\]
Invariant form of the above system is
\[
N_{-1,2}=-32M_{-1,2}^2+L_{-1,2},\quad P_{-1,2}=-M_{-1,2},\quad Q_{-1,2}=-\frac14(L_{-1,2}+1),\quad R_{-1,2}=2M_{-1,2}.
\]

\noindent {\bf Example~\refstepcounter{tbn}\thetbn\label{Jacobi-21}.} The weak Jacobi form $\varphi_{-2, 1}(\tau, z)$ of weight~$-2$ and index~$1$ is defined as
$\varphi_{-2, 1}(\tau, z)=\Delta^{-1/4}(\tau)\vartheta_1^2(\tau, z)$, see~\cite[formula~(4.29)]{DMZ}.
The function $f(\tau, z)=\varphi_{-2, 1}(\tau, z)$ satisfies the following $G_{-2, 1}$-invariant overdetermined involutive PDE system,
\begin{gather*}
2\pi if^2f_{\tau \tau \tau} = (2ff_{\tau \tau}-f_\tau^2)(2\pi if_\tau+f_{zz})-f_z^2f_{\tau\tau}-2ff_{\tau z}^2+2f_\tau f_zf_{\tau z},\\
f^2f_{\tau \tau z} = f_z(2ff_{\tau \tau}-f_\tau^2),\\
f^2f_{\tau zz} = 2\pi i f(ff_{\tau \tau}-f_\tau^2)+f_z(2ff_{\tau z}-f_\tau f_z),\\
f^2f_{zzz} = 4\pi if(ff_{\tau z}-f_\tau f_z)+f_z(2ff_{zz}-f_z^2),
\end{gather*}
which can be obtained from the system~(\ref{eqg}) by the change of variables
\[
f(\tau,z) = g^2(\tilde \tau,\tilde z),\quad \tau = \frac{i\tilde\tau}\pi,\quad z = \frac{\tilde z}{\pi}.
\]
Invariant form of the above system is
\[
N_{-2,1}=-32M_{-2,1}^2+2L_{-2,1}+1,\quad P_{-2,1}=0,\quad Q_{-2,1}=-\frac18(L_{-2,1}+1),\quad R_{-2,1}=M_{-2,1}.
\]

\noindent {\bf Example~\refstepcounter{tbn}\thetbn\label{JacobiTheta}.} The Jacobi theta function~$\vartheta_1$
is a Jacobi form of index~$\frac12$ and weight~$\frac12$. It satisfies the heat equation
\begin{subequations}\label{ThetaSys}
\begin{gather}\label{ThetaSysHeat}
4\pi i(\vartheta_1)_{\tau}=(\vartheta_1)_{zz},
\end{gather}
which coincides with equation (\ref{fSysHeat}) for $k=m=\frac{1}{2}$, as well as a sixth-order equation involving $z$-derivatives of $\vartheta_3$ only,
\begin{gather}\label{ThetaSysSix}
\left(\frac{(\ln \vartheta_1)_{zzzzz}}{(\ln \vartheta_1)_{zzz}}+12(\ln \vartheta_1)_{zz}\right)_z=0,
\end{gather}
\end{subequations}
which coincides with~(\ref{fSysSix}) for $k=\frac{1}{2}$. The system~\eqref{ThetaSys} is in involution and 
admits the Jacobi group~$G_{\frac12,\frac12}$ as a symmetry group.
It also holds for Jacobi theta functions~$\vartheta_2$, $\vartheta_3$ and~$\vartheta_4$ 
due to the fact that the transformations between different theta functions belong to the group~$G_{\frac12,\frac12}$.
In a somewhat different form, the system~\eqref{ThetaSys} for Jacobi theta functions was obtained in~\cite[Example 4.1]{CF}.
We also refer to~\cite{Brezhnev2013, Pavlov} for other differential systems satisfied by the Jacobi theta functions.

The system~\eqref{ThetaSys} can be written in invariant form as
\[
\mathcal L_{\frac12,\frac12}=0,\quad \mathcal M_{\frac12,\frac12}^3(\mathcal N_{\frac12,\frac12}-2\mathcal P_{\frac12,\frac12})-8=0.
\]

\noindent {\bf Example~\refstepcounter{tbn}\thetbn\label{Weierstrass}.} The Weierstrass $\wp$-function is a Jacobi form of weight~$2$ and index~$0$,
\[
\wp(\tau,z)=\frac1{z^2}+\sum\limits_{\omega\in L/\{0\}}\left(\frac1{(z-\omega)^2}-\frac1{\omega^2}\right),
\]
where the lattice $L=\mathbb Z+\mathbb Z\tau$. Using differentiation rules from~\cite[Formulae~(37), (100) and~(101)]{Brezhnev2013},
we find that the Weierstrass $\wp$-function satisfies the involutive $G_{2,0}$-invariant system of differential equations
\begin{gather*}
\wp_{\tau \tau} = \frac{\frac{1}{16}\phi+3\pi^2(2\wp\wp_{\tau z}^2-6\wp_\tau\wp_z\wp_{\tau z}+3\wp_\tau^2\wp_{zz})}{3\pi^2(2\wp\wp_{zz}-3\wp_z^2)},\\
\wp_{\tau zz} = \frac{\frac{i}{4}\phi+\pi(36\wp\wp_z(2\wp\wp_{\tau z}-3\wp_\tau\wp_z)-12\wp_{zz}(\wp_z\wp_{\tau z}-\wp_\tau\wp_{zz}))}{3\pi(2\wp\wp_{zz}-3\wp_z^2)}, \\
\wp_{zzz} = 12\wp\wp_z,\quad \text{where} \quad \phi:=16\wp_{zz}^3-72\wp^2\wp_{zz}^2 -216\wp\wp_z^2\wp_{zz}+54\wp_z^4+864\wp^3\wp_z^2.
\end{gather*}
Given $\wp_z^2=4\wp^3-g_2(\tau)\wp-g_3(\tau)$ where $g_2(\tau)=\frac{4}{3}\pi^4E_4(\tau)$ and $g_3(\tau)=\frac{8}{27}\pi^6E_6(\tau)$, we obtain~$\phi=-2(g_2^3-27g_3^2)=-2(2\pi)^{12}\Delta(\tau)$.
Note that the first of the above equations can be written in a symmetric Monge--Amp\`ere form,
\[
2\wp(\wp_{zz}\wp_{\tau \tau}-\wp_{z\tau}^2)=3\wp_z^2\wp_{\tau \tau}-6\wp_\tau\wp_z\wp_{\tau z}+3\wp_\tau^2\wp_{zz}-\frac{(2\pi)^{10}}{6}\Delta(\tau).
\]
The invariant form of the above system is
\begin{gather*}
\mathscr N_2=12\mathscr L_2,\quad \mathscr J_2=-4\pi i\mathscr I_2, \\
24\pi^2(2\mathscr M_2-3\mathscr L_2^2)\mathscr I_2=8\mathscr M_2^3-36\mathscr M_2^2-108\mathscr L_2^2\mathscr M_2+27\mathscr L_2^2(\mathscr L_2^2+16).
\end{gather*}

\section{Concluding remarks}

\begin{itemize}

\item The results of this paper can be extended to other types of modular forms  (such as Siegel modular forms, Picard modular forms, etc), namely, every modular form $f$ on a discrete subgroup $\Gamma$ of a Lie group $G$ should solve a nonlinear PDE system $\Sigma$ such that:
\begin{itemize}
\item system $\Sigma$ is involutive (compatible);
\item system $\Sigma$ is of finite type (has  finite-dimensional solution space);

\item system $\Sigma$ is $G$-invariant, furthermore, the Lie group $G$ acts on the solution space of $\Sigma$ locally transitively and with an open orbit (thus, the dimension of the solution space of $\Sigma$ equals ${\rm dim}\,G$);

\item the modular form $f$ is a generic solution of system $\Sigma$ ($f$ belongs to the open orbit), in particular, solution $f$ has discrete stabiliser $\Gamma$;

\item system $\Sigma$ is expressible via algebraic relations among differential invariants of a suitable action of $G$.

\end{itemize}
In the case of classical modular forms $f$ considered in this paper, we have: $\Gamma=\textrm{SL}(2, \mathbb Z)$, $G=\textrm{SL}(2, \mathbb R)$, and system $\Sigma$ is a third-order nonlinear $\textrm{SL}(2, \mathbb R)$-invariant ODE for $f$.

In particular, involutive differential systems for Siegel modular forms should be based on differential invariants of the symplectic group $\textrm{Sp}(2g)$. Some results in this direction are already available, thus, differential systems for theta constants were discussed in~\cite{O2} (genus $g=2$) and~\cite{Zudilin} (general~$g$).

\item Although modular forms provide {\it generic} solutions of the ODEs discussed in this paper, the same ODEs possess {\it non-generic} rational solutions which may also be of interest. Thus, as already mentioned in the introduction, the integrability condition for the Lagrangian density $u_xu_yf(u_t)$ is the fourth-order ODE~(\ref{E_13Ord4}) for~$f(\tau)$. The generic solution of this ODE is the Eisenstein series, $f(\tau)=E_{1,3}(\tau)$, however, it also possesses a simple non-generic solution $f(\tau)=\tau$ which corresponds to the Lagrangian density $u_xu_yu_t$ (with interesting properties, see~\cite{FKT}).

Same applies to differential systems for Jacobi forms.


\item Every third-order ODE with $G_k$ symmetry can be linearised by a standard procedure as discussed, e.g., in \cite{Clarkson}: the general solution $f(\tau)$ of any such equation can be represented parametrically as
\begin{equation}
\tau=\frac{\tilde w}{w}, \qquad f=\frac{w^k}{W^{k/2}}
\label{sub}
\end{equation}
where $w(s)$ and $\tilde w(s)$ are two linearly independent solutions of a second-order linear equation $w_{ss}+p w_s+q w=0$, and
$W=\tilde w_sw-w_s\tilde w$ is the Wronskian of~$w$ and~$\tilde w$.
Here the coefficients~$p(s)$ and~$q(s)$ depend on the third-order ODE and can be efficiently reconstructed. The details are as follows. Differentiating the second equation (\ref{sub}) with respect to $\tau$ using the relations $\frac{ds}{d\tau}=\frac{w^2}{W}$ and $W_s=-pW$, one obtains
\[
I_k=(2\pi i k)^2\delta \quad {\rm and} \quad J_k=(2\pi i k)^3\delta_s
\]
where $\delta=\frac{1}{2}p_s-q+\frac{1}{4}p^2$ and $I_k, J_k$ are the invariants from Section \ref{sec:invar}. Thus, for a given third-order ODE $F(I_k, J_k)=0$, one has to choose the coefficients $p(s),\, q(s)$ such that $F((2\pi i k)^2\delta,\, (2\pi i k)^3\delta_s)=0$.

This linearisation procedure leads to familiar parametrisations of modular forms by hypergeometric functions.

\

\end{itemize}

\section*{Acknowledgements}

We thank F. Cl\'ery, G. van der Geer, M. Pavlov, R.O. Popovych, A. Prendergast-Smith,  F.~Stromberg, A. Veselov, C. Wuthrich and V. Zudilin for useful discussions.
The research of SO was supported by the NSERC Postdoctoral Fellowship program.

\end{document}